\documentclass[conference,letterpaper]{IEEEtran}

\usepackage[utf8]{inputenc} 
\usepackage[T1]{fontenc}
\usepackage{url}
\usepackage{ifthen}
\usepackage{cite}
\usepackage[cmex10]{amsmath} 
\usepackage{amssymb,amsfonts,bbm,amsthm,bm,xcolor}
\usepackage{stmaryrd, mathtools}

\usepackage{pgfplots,tikz}  
\usepackage{subcaption}
\newtheorem{theorem}{Theorem}

\newtheorem{definition}{Definition}
\newtheorem{remark}{Remark}

\newcommand\blfootnote[1]{%
	\begingroup
	\renewcommand\thefootnote{}\footnote{#1}%
	\addtocounter{footnote}{-1}%
	\endgroup
}

\begin{document}
\title{Reed-Muller Codes for Quantum Pauli and\\ Multiple Access Channels} 

\author{%
  \IEEEauthorblockN{Dina Abdelhadi\IEEEauthorrefmark{2}\IEEEauthorrefmark{3}, Colin Sandon\IEEEauthorrefmark{1}\IEEEauthorrefmark{3},  Emmanuel Abbe\IEEEauthorrefmark{1}\IEEEauthorrefmark{2}, Ruediger Urbanke\IEEEauthorrefmark{2} }
  
  \IEEEauthorblockA{
  \IEEEauthorrefmark{1}Chair of Mathematical Data Science, EPFL 
    \\\{colin.sandon, emmanuel.abbe\}@epfl.ch}
    \IEEEauthorrefmark{2}School of Computer and Communication Sciences, EPFL

    \\\{dina.abdelhadi, rudiger.urbanke\}@epfl.ch
    \\\IEEEauthorrefmark{3} Equal first author contributors.}

\maketitle~\blfootnote{\textcopyright 2025 IEEE. Personal use of this material is permitted.  Permission from IEEE must be obtained for all other uses, in any current or future media, including reprinting/republishing this material for advertising or promotional purposes, creating new collective works, for resale or redistribution to servers or lists, or reuse of any copyrighted component of this work in other works.}

\begin{abstract}
Reed-Muller (RM) codes have undergone significant analytical advancements over the past decade, particularly for binary memoryless symmetric (BMS) channels. We extend the scope of RM codes development and analysis to multiple-access channels (MACs) and quantum Pauli channels, leveraging a unified approach. Specifically, we first derive the achievable rate region for RM codes on so-called Q-MACs, a class of MACs with additive correlated noise. This is achieved via a generalization of the bending and boosting arguments defined in \cite{abbe2023proof}. We then put forward a  connection between the rate region of these Q-MACs and quantum RM codes designed for Pauli noise channels. This connection highlights a universality property of quantum RM codes, demonstrating their rate-optimal performance across a range of channel parameters, rather than for a single Pauli channel.

\end{abstract}

\section{Introduction}

\subsection{Reed-Muller Codes}

Reed-Muller (RM) codes are a family of binary error-correcting codes such that $RM(r,m)$ has codewords consisting of the evaluations of $m$-variate polynomials of degree at most $r$ on $\mathbb{F}_2^m$~\cite{1057465,6499441,abbe2020reed}. Despite the currently open problem of achieving general efficient decoders, RM codes continue to garner significant attention due to several notable properties. These include the simplicity of their construction, their superior performance under maximum a posteriori (MAP) decoding compared to Polar codes~\cite{mondelli2014polar}, and their tri-orthogonality properties, which facilitate the construction of quantum RM codes capable of supporting transversal $T$-gates~\cite{rengaswamy2020optimality}.

RM codes have been proven to achieve the capacity of binary erasure channels~\cite{kudekar2016reed} and, more recently, of binary-input memoryless channels for both the bit \cite{reeves2021reed} and block \cite{abbe2023proof} error probability. These latter results also imply that it is  possible to construct ``quantum RM codes'' that achieve the hashing bound for Pauli channels. The results of \cite{abbe2023proof} on block error probability imply further that RM codes achieve strong secrecy up to capacity on classical-quantum channels for both the BSC and pure-state wire-tap channels.

\subsection{Contributions -- Executive Summary}
We study the application of RM codes to a two-user multiple-access channel (MAC) constructed from a memoryless channel with independently chosen binary inputs subjected to correlated noise, which we refer to as a Q-MAC. This model is relevant for classical communication and, perhaps more significantly, serves as an abstraction for quantum communication over Pauli channels.

Specifically, by identifying an achievable rate region for the Q-MAC, we determine regions of noise parameters where a single quantum code can facilitate reliable transmission. This approach proves a universality property of quantum RM codes, enabling robust performance across varying noise conditions.

It is important to note that the Q-MAC is a classical multiple access channel, with classical inputs and output, not to be confused with quantum multiple access channels. The `Q' in Q-MAC refers to the correspondence between the decoding problem in the classical multiple access setting and decoding problem in the case of the quantum Pauli channel. 
\subsection{Quantum Codes}

In our setting, quantum information is represented by the states of two-level systems, known as qubits, and the rate of a quantum code is defined as the ratio of logical qubits to physical qubits.

Quantum error-correcting codes enable the protection of quantum states against noise, playing a crucial role in quantum  communication and computation~\cite{Shor1995scheme}. 
Constructing quantum codes is a non-trivial task, due to various constraints dictated by the rules of quantum mechanics, including the no cloning theorem, and the fact that measurement disturbs the quantum state. The subclass of quantum CSS codes, introduced in~\cite{calderbank1996good,steane1996error}, provides a simplified framework for constructing quantum codes by leveraging existing classical codes. Given two classical binary linear codes $\mathcal{C}_X = [n, k_X]$ and $\mathcal{C}_Z = [n, k_Z]$, satisfying the condition $\mathcal{C}_X^\perp \subseteq \mathcal{C}_Z$~\cite{eczoo_qubit_css}, a valid quantum CSS code $\mathcal{C}_Q = \llbracket n, k_X + k_Z - n \rrbracket$ can be constructed. Since the parity checks of a quantum code correspond to quantum observables to be measured, these conditions ensure that these observables commute. 

\subsection{Quantum Channels}
Various mathematical models for quantum noise have been proposed in the literature, with one prominent model being that of Pauli noise~\cite{gottesman2009introductionquantumerrorcorrection}.

A quantum Pauli channel maps an input quantum state $\rho$ to itself with probability $p_I$, applies a Pauli $X$ error with probability $p_X$ (resulting in $X\rho X$), applies a Pauli $Y$ error with probability $p_Y$ (resulting in $Y\rho Y$), or applies a Pauli $Z$ error with probability $p_Z$ (resulting in $Z\rho Z$). Since the Pauli matrix $Y = iXZ$, a $Y$ error is equivalent to the simultaneous application of $X$ and $Z$ errors.

For quantum Pauli channels, the hashing bound provides an achievable rate for error-free transmission~\cite{bennett1996mixed}. It corresponds to the coherent information of a single channel use and is given by the expression $1 - H[p_I, p_X, p_Y, p_Z]$, where $H$ denotes the Shannon entropy of the probability distribution $\{p_I, p_X, p_Y, p_Z\}$.

Decoding a quantum CSS code over a Pauli channel can be mapped to the problem of decoding classical binary correlated errors, using a mapping of Pauli errors to 2-bit binary vectors, $I\rightarrow (0,0), X\rightarrow(1,0),Y\rightarrow (1,1), Z\rightarrow(0,1).$

\subsection{Quantum Codes based on RM Codes}
Quantum CSS RM codes, as well as non-CSS variants, have been proposed in~\cite{steane1996quantumreedmullercodes} and~\cite{zhang1997quantumreedmullercodes}. These works focused on constructing valid quantum RM codes,  utilizing the nesting and duality of RM codes, and analyzing their rate and distance properties:  Let $\mathcal{C}_X = RM(r_X, m)$ with rate $R_X$ and $\mathcal{C}_Z = RM(r_Z, m)$ with rate $R_Z$. The dual code is $\mathcal{C}_X^\perp = RM(m - r_X - 1, m)$ with rate $R_X^\perp = 1 - R_X$. The CSS condition $\mathcal{C}_X^\perp \subseteq \mathcal{C}_Z$ translates to $m \leq r_X + r_Z + 1$, or equivalently, $R_X + R_Z \geq 1$.

As discussed in~\cite[Section 1.5.2]{goswami2021quantum}, each Pauli channel corresponds to two classical binary memoryless (BMS) channels: an $X$-error channel $W_X = BSC(p_X + p_Y)$, and a $Z$-error channel $W_Z$. To account for the correlation between $X$ and $Z$ errors, we model $W_Z$ as a mixture of two binary symmetric channels (BSCs). The specific BSC realization of $W_Z$ depends on whether an $X$ error has occurred, with an additional output flag bit indicating the presence of an $X$ error. With probability $p_X + p_Y$ (when an $X$ error occurs), $W_Z$ acts as $\text{BSC}\big(\frac{p_X}{p_X + p_Y}\big)$. With probability $p_I + p_Z$ (when no $X$ error occurs), $W_Z$ acts as $\text{BSC}\big(\frac{p_Z}{p_I + p_Z}\big)$.   Decoding errors on these induced classical channels enables correcting Pauli errors on the quantum channel. A little thought shows that since RM codes achieve the capacity of these classical channels~\cite{reeves2021reed,abbe2023proof}, $I(W_X)$ and $I(W_Z)$, quantum RM codes achieve the hashing rate $I(W_X) + I(W_Z) - 1$ over the quantum Pauli channel. We will refer to this type of decoding, where one type of error is decoded first, followed by decoding the other code given the decoded information as \emph{successive decoding}.

\section{Definitions}

\subsection{Connection between MAC decoding and quantum Pauli channels}
\begin{definition}(Q-MAC)

A Q-MAC is a two-user memoryless MAC with a quaternary output alphabet $\mathcal{Y}\in \mathbb{F}_2^2,$
generated as $Y=(Y^{(1)},Y^{(2)})=(X^{(1)},X^{(2)})+(\Delta^{(1)},\Delta^{(2)}),$ with $X^{(i)} \in \mathbb{F}_2, \Delta^{(i)} \in \mathbb{F}_2$ for $i\in\{1,2\}$, and $(\Delta^{(1)},\Delta^{(2)}) \sim P_{\text{noise}}$. Note that $\Delta^{(1)}$ and $\Delta^{(2)}$ are in general correlated.   
\end{definition}
We show that (non-degenerate) decoding of errors occurring over codewords of a quantum CSS code transmitted over a quantum Pauli noise channel can be reduced to the problem of decoding a Q-MAC.

Since Pauli errors can be mapped to 2-bit binary vectors, we can express the noise model of a Pauli channel as a Q-MAC, such that the event of a Pauli error occurring on the Pauli channel corresponds to the event of a correlated 2-bit error on the Q-MAC. 
Due to the correspondence between the two models, the ability to decode errors that occur on a Q-MAC implies the ability to determine errors occurring on the quantum Pauli noise channel. 

\subsection{Contributions -- Extended Version}

Successive decoding of quantum RM codes can be seen as a technique for transmission at a rate pair achieving one of the corner points of the Q-MAC. Our aim is to investigate joint decoding suitable for transmission on further points of the dominant face of the Q-MAC other than the corner points. Although, in general, joint decoding might be more complex, it has the advantage of automatically exhibiting some universality property, allowing reliable transmission over a whole set of channel parameters. We will belabor this point when discussing Figure~\ref{fig:robustness}.

Explicitly, we determine necessary conditions on the two codes $C_1,C_2$ used by the two users on the Q-MAC for a rate pair $(R_1,R_2)$ to be achievable, and show that these conditions are also sufficient in case of Reed-Muller codes.
To that end, we define the following intermediate notions of decoding: 
 \begin{definition}(Merged weak local decoding)
Merged weak local decoding refers to decoding one of the two codeword coordinates or their XOR with error probability $1/2-\Omega(1)$ (the choice of the bit pair to decode is irrelevant for transitive codes).
\end{definition}
\begin{definition}(Global decoding)
Global decoding refers to decoding both codewords with error probability $o(1)$.
\end{definition}
Our proof proceeds in several stages:
\begin{itemize}
\item {\bf Bending:} A pair of $(R_1,R_2)-RM$ codes achieves \emph{merged weak local decoding} on a Q-MAC if $R_1+R_2$ is less than the channel's capacity.
\item {\bf Boosting:} If we can achieve \emph{merged local decoding} on a Q-MAC with a given pair of RM codes then given RM codes with slightly lower rates we can completely recover at least one of the codewords or their XOR with probability $1-o(1)$.
\item {\bf Reduction to a binary channel:} Given the value of one of the codewords or their XOR, decoding the other (or either in the XOR case) is equivalent to decoding it on a noisy channel. If $(R_1,R_2)$ is inside the MAC region and $\min(R_1,R_2)< I[X^{(1)}_0;X^{(1)}_0+X^{(2)}_0,Y_0]$ then the code in question has a rate that is less than the capacity of the channel, so we can decode the appropriate code.
\item {\bf MAC-to-Quantum map:} A pair of $(R_1,R_2)-RM$ codes achieves \emph{global decoding} on a quantum Pauli noise channel if both the \emph{reduction}, as well as the quantum CSS condition for commutativity of parity checks hold.
\end{itemize}

\section{Main Results}
\subsection{Necessary rate region for multiple access channel with overlapping codes}
First, let $P_{noise}$ be a probability distribution over $\mathbb{F}_2^2$. Now, we consider a problem where we have two linear codes $C_1$ and $C_2$ over $\{0,1\}^n$ of rates $R_1$ and $R_2$ respectively. Denoting vectors using bold font, we choose $\bm{X}^{(1)}\in C_1$ and $\bm{X}^{(2)}\in C_2$ uniformly at random. Then, for each $1\le i\le n$, we set $X_i=(X^{(1)}_i, X^{(2)}_i)$, draw $\Delta_i\sim P_{noise}$ and set $Y_i=X_i+\Delta_i$. This is a class of symmetric multiple access channels.

The key question is whether or not we can recover the value of $\bm{X}$  from the value of $\bm{Y}$. Theorem \ref{thm:necessary} determines the necessary conditions for recovering $\bm{X}$ from $\bm{Y}$.

\begin{theorem} \label{thm:necessary}
Consider the scenario explained above. Let $(X^{(1)}_0,X^{(2)}_0)$ be drawn uniformly at random from $\mathbb{F}_2^2$, $\Delta_0\sim P_{noise}$, and $Y_0=(X^{(1)}_0,X^{(2)}_0)+\Delta_0$. In order for it to be possible to recover $\bm{X}$  from $\bm{Y}$  with probability $\Omega(1)$, all of the following must be true:
\begin{enumerate}
\item $R_1+R_2\le I[(X^{(1)}_0,X^{(2)}_0);Y_0]$.
\item $R_1\le I[X^{(1)}_0;X^{(2)}_0,Y_0]$.
\item $R_2\le I[X^{(2)}_0;X^{(1)}_0,Y_0]$.
\item $\log_2(|C_1\cap C_2|)/n \le I[X^{(1)}_0;X^{(1)}_0+X^{(2)}_0,Y_0]$.
\end{enumerate}
\end{theorem}
\begin{proof}
The first of these conditions is necessary because otherwise one would be able to transmit an amount of information across the channel greater than its capacity. The other three follow from the fact that if one is able to recover {$\bm{X}$} from {$\bm{Y}$} it must also be possible to recover {$\bm{X}$} from $({\bm{Y}},{\bm{X}}^{(1)})$, $({\bm{Y}},{\bm{X}}^{(2)})$, or $({\bm{Y}},{\bm{X}}^{(1)}+{\bm{X}}^{(2)})$. The necessity of the second and third criteria follow immediately from this, whereas showing the necessity of the final criterion takes a little more work. 

To show that it is necessary, first observe that given the value of ${\bm{X}}^{(1)}+{\bm{X}}^{(2)}$ we can find some ${\bm{X}}^{\star (1)}$ and ${\bm{X}}^{\star(2)}$ with the appropriate sum. Then, there must exist some ${\bm{\delta}}^x\in C_1\cap C_2$ such that ${\bm{X}}^{(1)}={\bm{X}}^{\star(1)}+{\bm{\delta}}^x$ and ${\bm{X}}^{(2)}={\bm{X}}^{\star(2)}+{\bm{\delta}}^x$. So, recovering the value of $\bm{X}$ reduces to the problem of finding the value of ${\bm{\delta}}^x$. It is a random codeword in $C_1\cap C_2$, and the information provided on it by knowing the values of {$\bm{Y}$} and ${\bm{X}}^{(1)}+{\bm{X}}^{(2)}$ is equivalent to the information provided by running it through a channel that for each $i$ returns $(\delta^x_i,\delta^x_i)+\Delta_i$. This channel has a capacity of $I[X_0^{(1)};X_0^{(1)}+X^{(2)}_0,Y_0]$ and the desired conclusion follows.
\end{proof}
\begin{remark}
If one of $C_1$ and $C_2$ is contained in the other then $\lim_{n\to\infty} \log_2(|C_1\cap C_2|)/n=\min(R_1,R_2)$, allowing us to give a simpler alternative to the final criterion.
\end{remark}

\subsection{RM codes for multiple access channels}
If $C_1$ and $C_2$ are Reed-Muller codes, then the previous necessary conditions are essentially sufficient in the following sense.

\begin{theorem}\label{Thm:Achievability}
Consider the scenario explained above with $C_1$ and $C_2$ being Reed-Muller codes. Now, let $(X_0^{(1)},X_0^{(2)})$ be drawn uniformly at random from $\mathbb{F}_2^2$, $\Delta_0\sim P_{noise}$, and $Y_0=(X^{(1)}_0,X^{(2)}_0)+\Delta_0$ and assume that the following hold:
\begin{enumerate}
\item $R_1+R_2< I[(X^{(1)}_0,X^{(2)}_0);Y_0]$.
\item $R_1< I[X^{(1)}_0;X^{(2)}_0,Y_0]$.
\item $R_2< I[X^{(2)}_0;X^{(1)}_0,Y_0]$.
\item $\min(R_1,R_2)< I[X^{(1)}_0;X^{(1)}_0+X^{(2)}_0,Y_0]$.
\end{enumerate}
Then we can recover the value of $\bm{X}$ from $\bm{Y}$ with probability $1-o(1)$.
\end{theorem}

\begin{figure}
\includegraphics[width=0.5\textwidth]{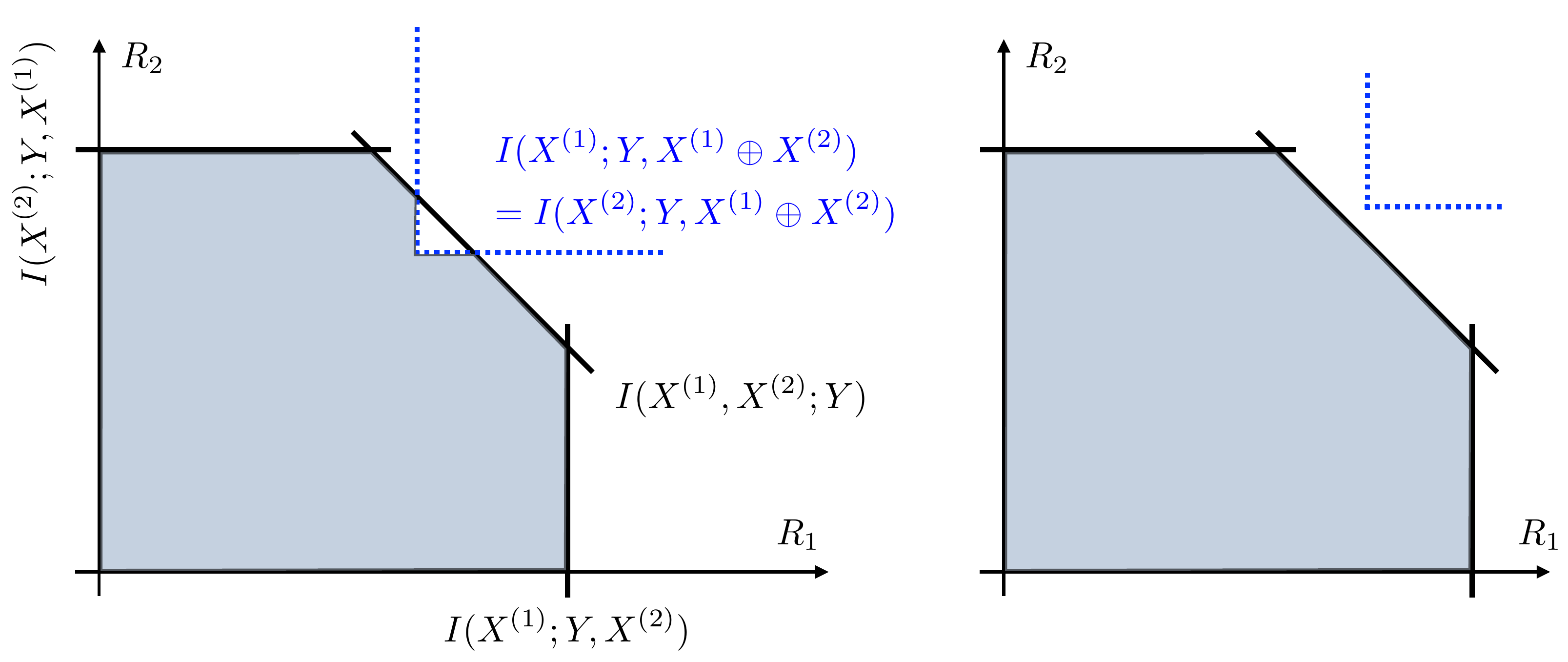}
    \caption{Achievable rate region of RM codes on a QMAC: the plain (black) bounds are the usual MAC bounds and the dash (blue) bounds result from the overlapping property of RM codes; the latter may be active (left) or inactive (right). }
    \label{fig:QMAC_rates}
\end{figure}  

\begin{proof}
First, define $m$, $r_1$ and $r_2$ so that $C_1=RM(r_1,m)$ and $C_2=RM(r_2,m)$, and the blocklength $n=2^m$. Let $m'=m-\epsilon\sqrt{m}$ where $\epsilon>0$ is a constant to be determined later. Next, let $C_1'=RM(r_1,m')$ and $C_2'=RM(r_2,m')$ . Then define $R'_1$, $R'_2$, $\bm{X'}$ etc analogously to $R_1$, $R_2$, $\bm{X}$ etc. Let ${m\choose \le r} \coloneqq \sum_{l=0}^{r} {m\choose l}$. Taking $m,m'$ as functions of the blocklength $n$, we can make $\lim_{n\to\infty} |R_1-R'_1|=\lim_{n\to\infty}2^{-(m-\epsilon\sqrt{m})}{m-\epsilon\sqrt{m}\choose \le r_1}-2^{-m}{m\choose \le r_1}$ and $\lim_{n\to\infty} |R_2-R'_2|= \lim_{n\to\infty}2^{-(m-\epsilon\sqrt{m})}{m-\epsilon\sqrt{m}\choose \le r_2}-2^{-m}{m\choose \le r_2}$ arbitrarily small by choosing a small enough value of $\epsilon$, so choose $\epsilon$ for which $R'_1+R'_2< I[(X^{(1)}_0,X^{(2)}_0);Y_0]$. 

Let ${\bm{Y'}_{-i}}$ denote the vector constructed from $\bm{Y'}$ by removing the $i^{th}$ bit, and ${\bm{Y'}_{<i}}\coloneqq (Y'_0,Y'_1,\dots, Y'_{i-1})$.
Next, observe that $H[\bm{Y'}]\le H[\bm{X'},\bm{Y'}]=H[\bm{X'}]+H[\bm{Y'}|\bm{X'}]=H[\bm{X'}]+(2-I[Y_0;(X^{(1)}_0,X^{(2)}_0)])n' \le (R'_1+R'_2+(2-R'_1-R'_2-\Omega(1))n'=2n'-\Omega(n')$. Thus there is an `entropy bending' argument (as in  \cite{abbe2023proof}) that 
\[\sum_{i=1}^{n'} H[Y'_i|\bm{Y'}_{<i}]=H[\bm{Y'}]=2n'-\Omega(n')\]
so there must exist $i$ for which $H[X'_i|{\bm{Y'}_{-i}}]=H[X'_i,\Delta'_i|{\bm{Y'}_{-i}}]-H[\Delta'_i]=H[Y'_i,\Delta'_i|\bm{Y'}_{-i}]-H[\Delta'_i]\le H[Y'_i|{\bm{Y'}_{-i}}]\le H[Y'_i|\bm{Y'}_{<i}]\le 2-\Omega(1)$. That implies that we must be able to determine the value of $X_i^{(1)\prime}$, $X^{(2)\prime}_i$, or $X_i^{(1)\prime}+X^{(2)\prime}_i$ from ${\bm{Y'}_{-i}}$ with nontrivial accuracy. By transitivity of RM codes, that then implies that this holds for all $i$, and that which of $X_i^{(1)\prime}$, $X^{(2)\prime}_i$, or $X_i^{(1)\prime}+X^{(2)\prime}_i$ we can determine with nontrivial accuracy must be the same for all $i$. 

The boosting argument~\cite{abbe2023proof} implies that if we can determine one bit of a codeword in $RM(r,m')$ with accuracy $1/2+\delta$ based on the noisy versions of the other bits of the codeword for some constant $\delta>0$, then we can completely recover the value of a codeword in $RM(r,m)$ from its noisy version with probability $1-o(1)$. $\bm{X}^{(1)}$, $\bm{X}^{(2)}$ and $\bm{X}^{(1)}+\bm{X}^{(2)}$ are codewords in $RM(r_1,m)$, $RM(r_2,m)$, and $RM(\max(r_1,r_2),m)$ respectively. So, if we can recover one bit of ${\bm{X'}^{(1)}}$, ${\bm{X'}^{(2)}}$, or $\bm{X'}^{(1)}+\bm{X'}^{(2)}$ with accuracy $1/2+\delta$ from the other bits in $\bm{Y'}$ then we can completely recover $\bm{X}^{(1)}$, $\bm{X}^{(2)}$, or $\bm{X}^{(1)}+\bm{X}^{(2)}$ from $\bm{Y}$ with probability $1-o(1)$. 

So, the fact that $R_1+R_2< I[(X_0^{(1)},X^{(2)}_0);Y_0]$ implies that we can recover at least one of $\bm{X}^{(1)}$, $\bm{X}^{(2)}$, and $\bm{X}^{(1)}+\bm{X}^{(2)}$ from $\bm{Y}$ with accuracy $1-o(1)$. Then the remaining assumptions imply that once we know that codeword recovering the remaining part of $\bm{X}$ is equivalent to recovering an RM code that was transmitted along a channel of capacity greater than its rate. We can do that with accuracy $1-o(1)$, so we can recover the value of $\bm{X}=(\bm{X}^{(1)},\bm{X}^{(2)})$ from $\bm{Y}$ with probability $1-o(1)$.
\end{proof}
Because of the previous section, the rate region in this theorem is tight. 
\subsection{Quantum Pauli channels}
Adding the condition for constructing a valid quantum RM code; $R_1+R_2\geq 1$ to the conditions in Theorem \ref{Thm:Achievability}, we can deduce a rate region where decoding the errors occurring in a Pauli noise channel modeled by a Q-MAC is possible. 
Usually for quantum error correction, syndrome decoding is used. However, our discussion involving codeword decoding for the Q-MAC is applicable using the Steane method for syndrome extraction \cite{gottesman2009introductionquantumerrorcorrection}, as this method generates a random codeword affected by the original noise vector, and decoding this random codeword is equivalent to finding the error that occurred on the original encoded quantum state. 

In Figure \ref{fig:robustness}, we fix a quantum RM code constructed out of two classical RM codes $C_1,C_2$ with $R_1=0.8$ and $R_2=0.8$, respectively, and we compare the region of channels $(p_X,p_Y,p_Z)$ that are decodable by successive decoding (Fig. \ref{fig:first}) with those that are decodable by joint decoding according to Theorem \ref{Thm:Achievability} (Fig. \ref{fig:second}). In other words, the plots of Figure \ref{fig:robustness} indicate the channels where a quantum RM code with $R_1=0.8,R_2=0.8$ fulfills the information theoretic conditions required for successful decoding. 
While Theorem~\ref{Thm:Achievability} provides the conditions for successful joint decoding, the conditions for successful successive decoding considered in Fig. \ref{fig:second} can be analyzed under three distinct settings:

\textbf{Setting I:} The decoder first decodes $X$ errors over $\text{BSC}(p_X + p_Y)$, requiring $R_1 = 0.8 \leq 1 - h_b(p_X + p_Y)$, where $h_b(\cdot)$ is the binary entropy function. Subsequently, it decodes $Z$ errors conditioned on the knowledge of $X$ errors, requiring $R_2 = 0.8 \leq I\left[X_0^{(2)};Y_0|X_0^{(1)}\right]= 1 - (p_X + p_Y) h_b\left(\frac{p_X}{p_X + p_Y}\right) - (p_I + p_Z) h_b\left(\frac{p_Z}{p_I + p_Z}\right)$.

\textbf{Setting II:} The decoder first decodes $Z$ errors over $\text{BSC}(p_Z + p_Y)$, requiring $R_2 = 0.8 \leq 1 - h_b(p_Z + p_Y)$. It then decodes $X$ errors conditioned on the knowledge of $Z$ errors, requiring $R_1 = 0.8 \leq I\left[X_0^{(1)};Y_0|X_0^{(2)}\right] =  1 - (p_Z + p_Y) h_b\left(\frac{p_Y}{p_Z + p_Y}\right) - (p_I + p_X) h_b\left(\frac{p_X}{p_I + p_X}\right)$.

\textbf{Setting III:} The decoder first decodes the XOR of the errors, $\bm{\Delta}^{(1)} + \bm{\Delta}^{(2)}$, which corresponds to a $\text{BSC}(p_X + p_Z)$. Since $C_1$ and $C_2$ are Reed-Muller (RM) codes with the same blocklength, one must be a subcode of the other, implying that their XOR is a codeword in the larger code. This imposes the condition $\max(R_1, R_2) = 0.8 \leq 1 - h_b(p_X + p_Z)$. After decoding the XOR, the decoder proceeds to decode either $\bm{\Delta}^{(1)}$ or $\bm{\Delta}^{(2)}$, requiring $\min(R_1, R_2) = 0.8 \leq I\left[X_0^{(1)}; Y_0 \mid X_0^{(1)} + X_0^{(2)}\right] = 1 - (p_I + p_Y) h_b\left(\frac{p_Y}{p_I + p_Y}\right) - (p_Z + p_X) h_b\left(\frac{p_X}{p_Z + p_X}\right)$.

If a channel $(p_X,p_Y,p_Z)$  satisfies the conditions of any one of the three settings, then it is decodable by successive decoding. The plots focus on the low noise regime with $p_X, p_Y,p_Z \leq 0.05$. Observe that for successive decoding there are only one-dimensional regions (lines) where decoding is possible and the hashing bound ($R_1+R_2\leq I[X_0^{(1)},X_0^{(2)};Y_0]$) is fulfilled with equality, whereas for joint decoding this ``optimal region'' is two-dimensional, i.e., it is much larger.

This distinction is even more apparent in Figure \ref{fig:third}, where we show the cross-section corresponding to $p_X=p_Y$. The gray line shows the hashing bound. For all parameters below this hashing bound a decodable code can be constructed with rate $R_1+R_2-1=0.6$. The shaded region shows the parameters for which a quantum RM code comprised of two component codes of rates $R_1=R_2=0.8$ will succeed if decoded successively. This region touches the hashing bound only at two points. This is in stark contrast to the set of parameters for which the same quantum RM code succeeds if decoded jointly. This code is now ``optimal'', i.e., touches the hashing bound over a whole continuous region.
\begin{figure}
    \centering
\begin{subfigure}{0.4\textwidth}
    \includegraphics[width=\textwidth]{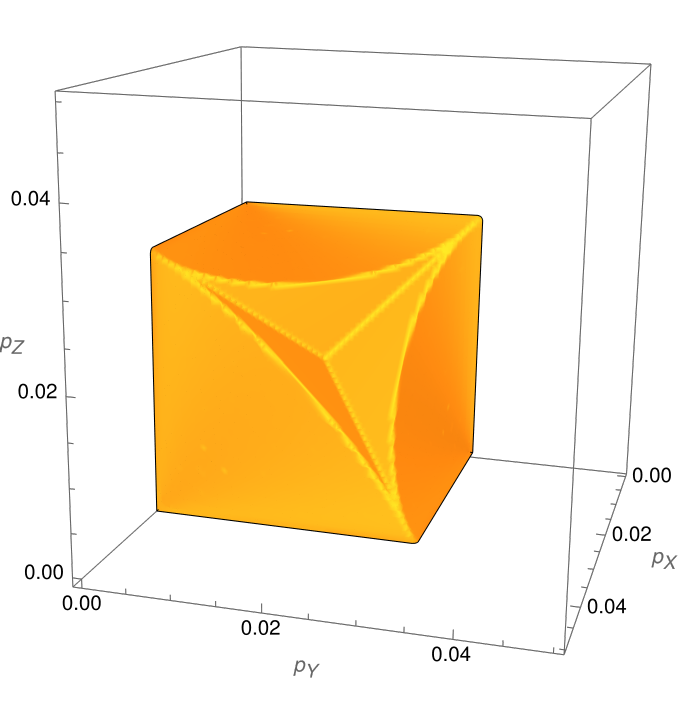}
    \caption{Region of channels decodable by successive decoding. Note that this region is the union of the three possible successive decoding strategies: (i) decode the $X$ error (i.e., $\bm{\Delta}^{(1)}$), then use the decoded information to decode the $Z$ error ($\bm{\Delta}^{(2)}$); (ii) or decode $Z$ errors first then $X$, and finally, (iii) we may instead first decode $\bm{\Delta}^{(1)} + \bm{\Delta}^{(2)}$, then decode either of $\bm{\Delta}^{(1)} $ or $\bm{\Delta}^{(2)}$ .}
    \label{fig:first}
\end{subfigure}
\hfill
\begin{subfigure}{0.4\textwidth}
    \includegraphics[width=\textwidth]{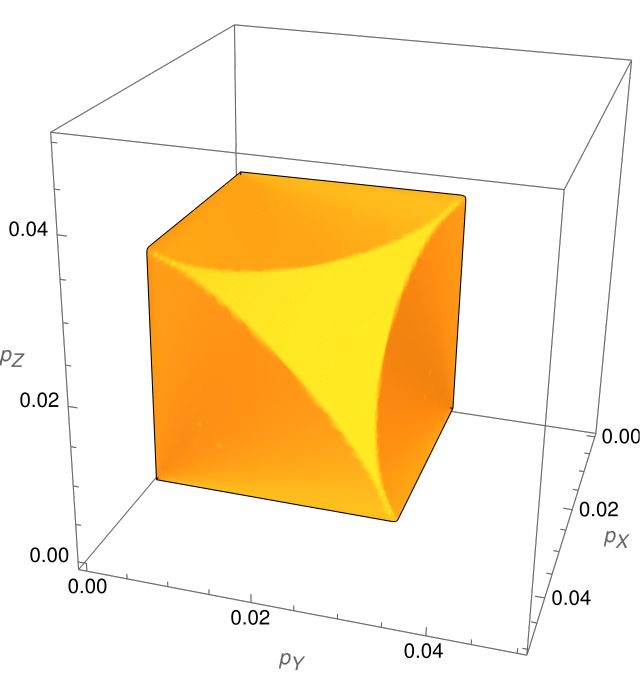}
    \caption{Region of channels decodable by joint decoding, where $(p_X,p_Y,p_Z)$ satisfy the conditions of Theorem \ref{Thm:Achievability}.}
    \label{fig:second}
\end{subfigure}    
\hfill

\caption{$R_1 =R_2 = 0.8$, focusing on the low noise regime $p_X,p_Y,p_Z\leq 0.05$}
    \label{fig:robustness}
\end{figure}
\begin{figure}
\centering
    \includegraphics[width=0.4\textwidth]{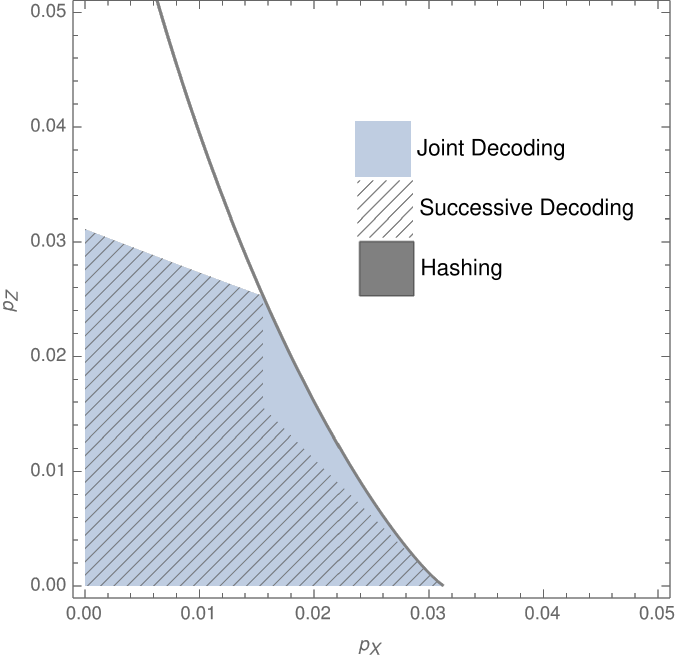}
    \caption{Cross-section $p_X=p_Y$ of Figure \ref{fig:robustness}.}
    \label{fig:third}
\end{figure}  
\section{Discussion}

\subsection{Variants of RM codes with augmented rate regions}

To the degree that the requirement that $\min(R_1,R_2)< I[X_0^{(1)};X_0^{(1)}+X_0^{(2)},Y_0]$ is a limitation, we have some options to potentially mitigate it. 

First of all, there is the  option to rely on time-sharing between codes with high $R_1$ and low $R_2$ and codes with low $R_1$ and high $R_2$, which could have been applied to the corner points directly.

Another option that prevents time sharing is to use the following `tensor products of RM codes' instead of RM codes themselves. Consider having $r_1$ and $r_2$ such that $RM(m,r_1)$ and $RM(m,r_2)$ have rates of $R_1$ and $R_2$ respectively. In this case, instead of setting $C_1$ and $C_2$ equal to $RM(m,r_1)$ and $RM(m,r_2)$ respectively we could set $C_1=RM(m,r_1)\otimes RM(m,m)$ and $C_2=RM(m,m)\otimes RM(m,r_2)$. In other words, $C_1$ would be the code whose codewords are the lookup tables for every polynomial in $\mathbb{F}_2^{2m}$ with degree in its first $m$ variables at most $r_1$ and $C_2$ would be the code whose codewords are the lookup tables for every polynomial in $\mathbb{F}_2^{2m}$ with degree in its last $m$ variables at most $r_2$. In this case, $C_1$ has rate $R_1$, $C_2$ has rate $R_2$, and $C_1\cap C_2=RM(m,r_1)\otimes RM(m,r_2)$ which has rate $R_1R_2$. That changes the requirement from $\min(R_1,R_2)< I[X_0^{(1)};X_0^{(1)}+X_0^{(2)},Y_0]$ to $$R_1R_2< I[X_0^{(1)};X_0^{(1)}+X_0^{(2)},Y_0].$$

A third option could be to take some of the coefficients of terms in $C_1\cap C_2$, change the value of $C_1$ or $C_2$ in order to require that they be $0$, and increase the value of $r_1$ and/or $r_2$ in order to compensate for the reduction in the rates of the codes. This would allow one to reduce $|C_1\cap C_2|$ while keeping $R_1$ and $R_2$ the same, but it has the problem that setting some coefficients to $0$ would likely disrupt the symmetries of the code. As such, the standard argument about our ability to recover noisy codewords would no longer apply and it is unclear how or if we could fix it.

A fourth option is to consider $\pi C_1$ and $C_2$, for some permutation $\pi$.

While these approaches help for MAC applications, except for the first, they are likely to result in codes that no longer satisfy the CSS condition, which means that they would not be usable for Pauli channels without additional modifications or restrictions.

\section*{Acknowledgment}
Drafts written by the authors were copy-edited for language using ChatGPT (v 4.0, OpenAI). The tool was limited to surface changes (grammar, spelling, style); it did not create content, ideas, data, analyses or references or alter scientific content. All corrections were checked by the authors, who approve the final text. 
\bibliographystyle{IEEEtran}
\bibliography{refs.bib}
\end{document}